\documentclass[letterpaper, 10 pt, conference]{ieeeconf}  % Comment this line out if you need a4paper

\IEEEoverridecommandlockouts                              % This command is only needed if 
\usepackage{cite}
\usepackage{graphicx,float}      % include this line if your document contains figures

\graphicspath{Figures/}
\DeclareGraphicsExtensions{.png,PNG,.pdf,.PDF}
\usepackage{bbm}
\usepackage{bm}
\usepackage{amsmath} % assumes amsmath package installed
\protected\edef\ell{\noexpand\ensuremath{{\mathchar\the\ell}}}
\usepackage{amssymb}  % assumes amsmath package installed
\usepackage{mathrsfs}
\usepackage{multicol}
\usepackage{xcolor}
\usepackage{subcaption}
\usepackage{mathtools}
\usepackage{soul}
\usepackage[T1]{fontenc}
\newtheorem{rem}{Remark}
\newtheorem{assum}{Assumption}
\newtheorem{defn}{Definition}
\newtheorem{prop}{Proposition}
\newtheorem{prob}{Problem}

\newtheorem{lemma}{Lemma}
\allowdisplaybreaks
\usepackage{eucal}
\usepackage{url}
\usepackage{comment}
\DeclareMathAlphabet{\pazocal}{OMS}{zplm}{m}{n}

\usepackage{amsmath}
\DeclareMathOperator*{\argmin}{arg\,min}

\newcommand{\cG}{\pazocal{G}}
\newcommand{\cE}{\pazocal{E}}
\newcommand{\cV}{\pazocal{V}}
\newcommand{\until}[1]{\{1,\dots, #1\}}

\newcommand{\danilo}[1]{{\color{red} Danilo: #1}}

\renewcommand{\baselinestretch}{0.99}
\title{\LARGE \bf
Model Predictive Control for Multi-Agent Systems under Limited Communication and Time-Varying Network Topology
}

\author{Danilo Saccani, Lorenzo Fagiano, Melanie N. Zeilinger and Andrea Carron% <-this % stops a space
\thanks{This research has been supported by the Italian Ministry of University and Research (MIUR) under the PRIN 2017 grant n. 201732RS94 ``Systems of Tethered Multicopters'' and by the Swiss National Science Foundation under the NCCR Automation (grant agreement 51NF40\_180545).}% <-this % stops a space
\thanks{D. Saccani is with the Institute of Mechanical Engineering, Ecole Polytechnique Fédérale de Lausanne (EPFL), CH-1015 Lausanne, Switzerland. (email: \tt\small {danilo.saccani@epfl.ch) } }%
\thanks{L. Fagiano is with the Dipartimento di Elettronica, Informazione e Bioingegneria, Politecnico di Milano, Piazza Leonardo da Vinci 32, Milano, Italy. (email: \tt\small {lorenzo.fagiano@polimi.it) } }%
\thanks{M. N. Zeilinger and A. Carron are with the Institute for Dynamic Systems and Control, ETH Zurich, Switzerland. (email {\tt\small \{ carrona, mzeilinger\}@ethz.ch)}}%
}

\begin{document}

\maketitle
\thispagestyle{empty}
\pagestyle{empty}
%%%%%%%%%%%%%%%%%%%%%%%%%%%%%%%%%%%%%%%%%%%%%%%%%%%%%%%%%
%%%%%%%%%%%%%%%%%%%%%%%%%%%%%%%%%%%%%%%%%%%%%%%%%%%%
\begin{abstract}                % Abstract of not more than 250 words.
In control system networks, reconfiguration of the controller when agents are leaving or joining the network is still an open challenge, in particular when operation constraints that depend on each agent's behavior must be met. Drawing our motivation from mobile robot swarms, in this paper, we address this problem by optimizing individual agent performance while guaranteeing persistent constraint satisfaction in presence of bounded communication range and time-varying network topology. The approach we propose is a model predictive control (MPC) formulation, building on multi-trajectory MPC (mt-MPC) concepts.
To enable plug and play operations when the system is in closed-loop without the need of a request, the proposed MPC scheme predicts two different state trajectories in the same finite horizon optimal control problem. One trajectory drives the system to the desired target, assuming that the network topology will not change in the prediction horizon, while the second one ensures constraint satisfaction assuming a worst-case scenario in terms of new agents joining the network in the planning horizon.
Recursive feasibility and stability of the closed-loop system during plug and play operations are shown.
The approach effectiveness is illustrated with a numerical simulation.
\end{abstract}

%\begin{keyword}
%Predictive control; Multi-agent systems; Coordination of multiple vehicle systems; Nonlinear predictive control;Networked systems; Nonlinear cooperative control; Control under communication constraints (nonlinearity);
%\end{keyword}

%===============================================================================

\section{Introduction}
The interest in autonomous mobile robots is ever increasing for applications \cite{siegwart2011introduction} ranging from military technology \cite{patil2020survey} to self-driving vehicles \cite{bagloee2016autonomous}.
In particular, multi-agent motion planning has proven successful due to its relevance for numerous real-life applications, see for example \cite{shamma2008cooperative}.
Among the different approaches for dynamic path planning, optimization-based ones, such as Model Predictive Control (MPC), see \cite{rawlings2017model}, have received broad attention thanks to their ability to manage state and input constraints while minimizing multi-objective cost functions.
%When multiple agents are able to interact and communicate with each other, they can share information and cooperate to jointly solve tasks.
When the communication between multiple agents depends on the agent's state, the generated communication network is time-varying, and at each time step, subsystems can leave or join the network.
%, subject to constant changes in terms of subsystems that are added or removed. 
The problem of efficiently treating agents or nodes joining or leaving a network has been referred to as a Plug and Play (PnP) problem in the literature \cite{stoustrup2009plug}.
In order to give guarantees on stability and constraint satisfaction for the new network topology, available results resort on offline re-design of the local controllers to accept a plug-in request.
The problem of automatic PnP is still an open challenge \cite{stoustrup2009plug} when a subsystem is added or removed without a request.
In this paper, we address the problem of autonomously navigating a group of robots to a target while guaranteeing collision avoidance despite plug-in plug-out operations. Each agent is able to communicate with neighbouring robots.
Neighbouring subsystems are defined based on the agent's current state, and, during the navigation, the agent's state evolves and shifts its communication capabilities. Thus, the network topology can not be enforced to remain the same, but it is intrinsically time-varying and evolves during navigation. 
This time-varying nature of the network topology calls for an approach that must be able to tolerate plug-in/out operations without a-priori requests.
The presented solution is based on the multi-trajectory MPC concept firstly introduced in \cite{saccani2021autonomous} and nonlinear tracking MPC proposed in~\cite{limon2018nonlinear,fagiano2013generalized}.
Specifically, the multi-trajectory formulation  %where each agent predicts two state trajectories 
is used to balance two potentially conflicting requirements: tracking of the target and safe behaviour in case of network topology changes.
\subsubsection*{Related work}
Numerous applications necessitate enhancing performance without compromising safety. This problem has been addressed in different works, often exploiting optimization in order to satisfy safety constraints.
%%% SAFETY FILTERS
In \cite{wabersich2018safe}, the authors derived a predictive safety filter to ensure the system's safety while an external, potentially unsafe, learning-based control action optimizes the system's performance. The same concept has been applied to distributed networked systems in \cite{muntwiler2020distributed}.
%%% CONTROL BARRIER FUNCTION
Also control barrier function theory has been investigated to guarantee the system's safety \cite{ames2019control}.
In \cite{wabersich2022predictive}, a soft-constrained predictive control problem has been used as a recovery mechanism for a safety filter to guarantee the feasibility of the problem.
%%% mt-MPC
These approaches guarantee the system's safety but rely on an external controller to maximize the performance. Moreover, predictive safety filters are designed to handle uncertainty in system dynamics, but not changing network topologies.
To combine performance and safety in a single approach, in \cite{tordesillas2021faster}, the authors proposed the use of multiple trajectories in a trajectory planner where a back-up trajectory is used to ensure safety.
In \cite{saccani2021autonomous}, \cite{soloperto2022safe} and \cite{saccani2022}, the approach has been considered in an MPC framework providing theoretical guarantees on robust constraint satisfaction and convergence of the approach.
In \cite{alsterda2021contingency}, a similar concept has been exploited to trade-off the behaviour of a nominal with that of a contingency model to control a self-driving car.

%%%% PNP
%When a network setup is considered, it's not only important to ensure safety in the multi-agent system with a given number of agents but also to account for network topology changes. 
When a network setup is considered, ensuring safety while accounting for network topology changes is important. 
To this aim, a significant effort has been made to address the PnP problem. In \cite{zeilinger2013plug}, the authors present a transition scheme that prepares the system for the new network topology. The plug-in plug-out requests are elaborated by the network, and if the request can be accepted, a re-design of local controllers is performed. These results have been exploited in \cite{carron2021plug} to derive a safety filter able to provide safety verification during plug-in plug-out operations when a distributed learning-based control action is applied to the system.
Similarly to the approach presented in \cite{zeilinger2013plug}, in \cite{riverso2013plug}, an offline re-design of local controllers has been proposed for the plugging-in plugging-out of a subsystem when the network accepts the request.
In contrast, the approach proposed in this work designs a safe, feasible trajectory online that can always tolerate possible plug-in plug-out operations deriving from the time-varying network topology.
\subsubsection*{Contributions}  
The main contributions of this paper are twofold: the first is a safe control scheme for multi-agent systems ensuring collision avoidance with the current neighbouring agents using multi-trajectory MPC. The second contribution is to enable automatic plug-and-play operations in a time-varying network topology of agents with limited communication capabilities that, differently from other works in literature, cannot be denied.
%Each agents is able to communicate in a bidirectional way to agents close enough to its position leading to a time-varying communication topology defined based on agents’ current state.
%The proposed approach, amenable to distributed computation, considers two trajectories in the same FHOCP where a safe trajectory, accounting for a worst-case scenario, allows safe plug-and-play operations due to topology changes without request.
%%%%%%%%%%%%%%%%%%%%%%%%%%%%%%%%%%%%%%%
\section{Problem description}
In this section, we first introduce the system setup, discuss the communication model among agents, and the resulting communication network. Finally, we will state the problem we aim to solve.
\subsection{System setup}
We consider a group of mobile agents where each agent is identified by an integer $i\in\pazocal{M}=\{1,\dots,N_a\}$ and behaves according to the following discrete-time nonlinear dynamics
\begin{align}\label{eq:dt-system}
    x_i(k+1)&=f_i(x_i(k),u_i(k))  \\
    %\textcolor{red}{y_i(k)}&=h_i(x_i(k),u_i(k)) \\    
    p_i(k)&=C_ix_i(k), \nonumber
\end{align}
where $x_i(k)\in\mathbb{R}^{n_i}$ is the state vector, $u_i(k) \in\mathbb{R}^{m_i}$ the input vector, $f_i:\mathbb{R}^{n_i}\times \mathbb{R}^{m_i}\rightarrow \mathbb{R}^{n_i}$, and $C_i \in\mathbb{R}^{3 \times n_i}$ is a matrix that extracts the position $p_i(k)\in\mathbb{R}^3$ of the robot. We assume that each vehicle is able to measure its whole state~$x_i$. We denote with $(\bar{x}_i, \ \bar{u}_i)$ an equilibrium of system~\eqref{eq:dt-system} and we consider a state reference tracking problem where $r_i=f_i(\bar{x}_{r,i},\bar{u}_{r,i})$ is the constant state reference of the $i$-th agent. Furthermore, we assume that $f_i(x_i(k),u_i(k))$ is differentiable at every equilibrium point and the linearized model is controllable.
Let us consider, without loss of generality, that $p_i(k)$ is at the top of the state vector $x_i(k)$ and introduce the operator $$\phi (p_i) = [p_i^T,0, \dots, 0]^T\in\mathbb{R}^{n_i}$$ %that generates a vector of the dimension of $x_i(k)$ with first entry, corresponding to the position $p_i$ of the $i$-th robot.
\begin{assum} \label{ass:posInv}
The vehicles' dynamics~\eqref{eq:dt-system} are position invariant, i.e. $\forall p_j \in \mathbb{R}^3,\ x_i(k),\ u_i(k), \ f_i(x_i(k)+\phi(p_j),u_i(k))=x_i(k+1)+\phi(p_j)$.
\end{assum}
\begin{rem}As an example, Assumption \ref{ass:posInv} is satisfied when the position is the output of an integrator, %and does not affect the input to the same integrator, 
which is a typical condition in autonomous vehicles.\end{rem}
Finally, each agent is subject to convex time-invariant state and input constraints of the form 
\begin{equation}
    (x_i, u_i) \in \pazocal{X}_i\times\pazocal{U}_i, \ \forall i\in\pazocal{M}.
\end{equation}
\subsection{Communication and network topology}
\begin{figure}
	\centering
	\setlength\belowcaptionskip{-1.2\baselineskip}\includegraphics[width=.9\columnwidth]{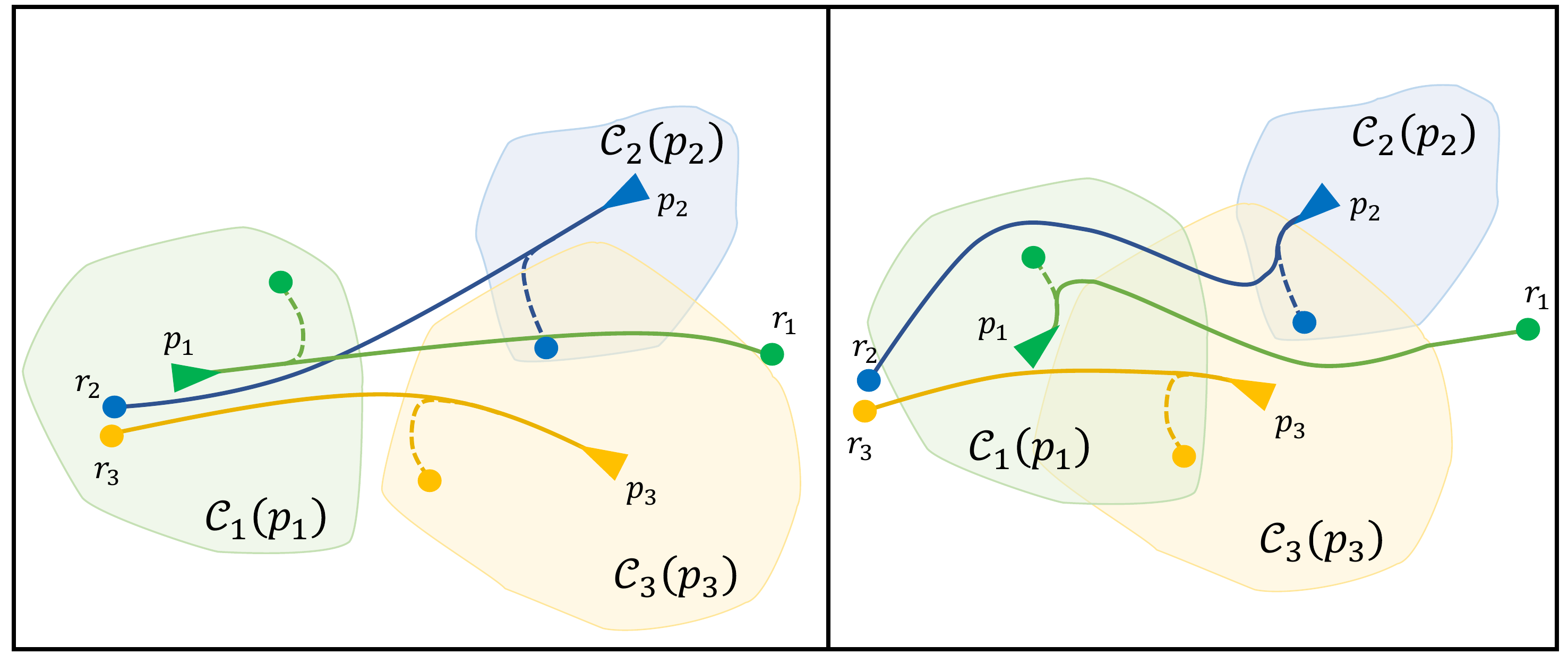}
	\caption{Plug and play multi-trajectory MPC for three agents whose position is represented by coloured triangles `$\vartriangle$'. In the left figure, the system is at time instant $k$, while in the right one, the system is at time $k+1$.  $\pazocal{C}_i(p_i)$ are the communication sets and $r_i$ the desired position references. Coloured lines: tracking trajectories; dashed lines: safe trajectories. Coloured dots are the final positions of by the various trajectories.}
	\label{fig:mt-MPC}
\end{figure}
%In this subsection, we present the considered communication and the obtained network topology. \\
Each agent is equipped with a communication system e.g., an antenna, characterized by a communication set $\pazocal{C}_i(p_i(k))=[p_i(k) \oplus \pazocal{D}_i]$, where $\pazocal{D}_i\subseteq\mathbb{R}^3$ is a constant compact convex set centered at the origin
and $\oplus$ is the Minkowski sum.
In some practical applications, the communication set may be originally non-convex, %for example, due to obstacles blocking the communication. %\danilo{This statement could be misunderstood since it could seem a time-varying set. The point highlighted by reviewer 5, about the position in the set is interesting but the set $\pazocal{D}_i$ could be always be approximated as $\bar{\pazocal{D}}_i \ominus v_{max}T_s$, where $v_{max}$ is the maximum velocity that any agents can reach}. 
in these cases, one can still take a convex under-approximation of the communication region. Let us consider the following assumption. %for the communication between agents.
\begin{assum}\label{ass:communication}
Agent $i$ is able to communicate in a bidirectional way with agent $j$ when
\begin{equation}
    \pazocal{C}_i(p_i(k))\cap \pazocal{C}_j(p_j(k)) \neq \emptyset.
\end{equation}
 
 \smallskip 
 
%Furthermore, we assume that when the communication sets of two agents intersect, the topology of the network changes.
\end{assum} 
The dependence of the communication sets on the system position leads to a time-varying communication topology, including the case where there is no communication among the agents.
Formally, such a topology could be described as a time-varying graph of which we can consider the connected sub-graphs. We denote each of the connected sub-graph as a cluster~$\cG_m(k)=(\cV_m(k),\cE_m(k))$ with $m=\{1,\dots,N_c(k)\}$, where the set of nodes~$\cV_m(k)\in\until{N_a}$ represents the agents in the cluster, and the set of edges~$\cE_m(k) \subset\cV_m(k) \times \cV_m(k)$ contains the pairs of agents~$\{i,j\}$, which can communicate with each other at time~$k$.
%For each cluster , an undirected communication sub-graph .
Thus, at each time step, the agents are grouped in a time-varying number of clusters $N_c(k)$, hence the sum of cardinalities of the set of nodes is equal to the total number of agents $N_a$ and the number of clusters is equal to the number of connected sub-graphs at time $k$. For each cluster $\pazocal{G}_m(k)$, by combining the local system dynamics in~\eqref{eq:dt-system}, the nonlinear dynamics of the \textit{cluster} system is $x(k+1)=f(x(k),u(k))$, where $x(k)=$col$_{i\in\pazocal{V}_m(k)} (x_i(k))$, $u(k)=$col$_{i\in\pazocal{V}_m(k)} (u_i(k))$ and we can define $p(k)=$col$_{i\in\pazocal{V}_m(k)} (p_i(k))$.
Fig. \ref{fig:mt-MPC} shows an example with three agents at two subsequent time steps.
At time step $k$ (on the left), the agents on the right are able to communicate generating the cluster $\pazocal{G}_1(k)$ with cardinality of the of nodes' set $|\pazocal{V}_1(k)|=2$. Instead, the agent on the left cannot communicate with the others, representing a cluster $\pazocal{G}_2(k)$ with nodes' cardinality $|\pazocal{V}_2(k)|=1$.
At the subsequent time step (on the right), agents $1$ and $2$ can communicate each with agent $3$, thus defining a new communication topology with only one cluster $\pazocal{G}_1(k+1)$ with cardinality of the of nodes' set $|\pazocal{V}_1(k+1)|=3$. 
Thus, for a cluster $m$, we have a plug-in operation when $|\pazocal{V}_m(k)|<|\pazocal{V}_m(k+1)|$ and a plug-out one when $|\pazocal{V}_m(k)|>|\pazocal{V}_m(k+1)|$.
We finally define a position dependent set of neighbouring systems for each agent in the considered cluster.
\begin{defn}[Neighboring systems]\label{def:NeighboringSys}
For each cluster $m$ with $m=\{1,\dots,N_c(k)\}$, %agent $j$ is a neighbor of agent $i$ if $\pazocal{C}_i(p_i(k))\cap \pazocal{C}_j(p_j(k)) \neq \emptyset$. L
let us denote the set of all neighbors of $i\in\cV_m(k)$, including $i$ itself as $\pazocal{N}_i(k)=\{i\} \cup \{j:\{i,j\}\in\cE_m(k) \}$. The states of all vehicles $j\in\pazocal{N}_i(k)$ are denoted as $x_{\pazocal{N}_i(k)}=$col$_{j\in\pazocal{N}_i(k)} (x_j)\in\mathbb{R}^{n_{\pazocal{N}_i(k)}}$, where col denote a vector which consists of the stacked
subvectors $x_j$.
\end{defn}
%For the sake of notational simplicity, we omit the dependence of $\pazocal{N}_i$ from the time $k$.

\subsection{Collision avoidance}
To model the requirements of collision avoidance among agents, let us define the obstacle avoidance non-convex coupling constraint between neighboring agents as:
\begin{equation}\label{eq:oa_constraint}
    h_i(x_{\pazocal{N}_i(k)})\leq 0, \ \forall i\in\pazocal{V}_m(k),   \ \forall m\in \mathbb{N}_1^{N_c(k)},
\end{equation}
where $\mathbb{N}_a^b=\{ n\in\mathbb{N} \ |\ a\leq n\leq b\}$.
%For example, by defining with $\pazocal{Q}(x_i(k))$ the space occupied by the vehicle $i$ at time $k$, we want to guarantee that $\pazocal{Q}(x_i(k))\cap \pazocal{Q}(x_j(k)) = \emptyset$, $\forall j \in \pazocal{N}_i(k)$. 
Possible choices for this constraint will be shown in Section~\ref{ss:collisionavoid}.
Each agent has to avoid collisions with neighbouring agents satisfying constraint~\eqref{eq:oa_constraint} $\forall k\geq 0$.
Thus, due to the time-varying nature of the communication topology %defined by Assumption~\ref{ass:communication} and Definition~\ref{def:NeighboringSys}
, each agent must be able to tolerate a possible variation of the neighboring systems set $\pazocal{N}_i(k)$ guaranteeing the satisfaction of constraint~\eqref{eq:oa_constraint}.
\subsection{Problem formulation}
We are now in position to state the following problem.
\begin{prob} \label{pr:feedbackLaw}
Consider $N_a$ mobile robots with dynamics~\eqref{eq:dt-system}, subject to local state and input constraints $(x_i(k), u_i(k))\in \pazocal{X}_i\times\pazocal{U}_i$, $\forall i\in\pazocal{M}$.
Each agent can communicate with neighbouring agents according to Assumption~\ref{ass:communication} defining a time-varying number of clusters ranging in the time-varying set $\{ 1,\dots,N_c(k)\}$.
Each cluster presents a time-varying network topology as described in Definition~\ref{def:NeighboringSys}. 
We aim to design a state feedback control law that drives every agent to their reference $r_i$ or to the closest feasible steady state, %whose meaning will be rigorously defined later on, 
while avoiding collision with neighbouring agents satisfying constraint~\eqref{eq:oa_constraint}, $\forall k\geq 0$, despite the time-varying nature
of the communication topology and of the neighboring systems.
\end{prob}
%In other words, we aim to design a state feedback control law, composed by the control laws computed in each cluster, by communicating only with neighbouring agents.
%The control law must allow every agent to reach a predefined goal location in space avoiding collisions with all the other agents despite possible plug-in plug-out operations of new agents in the current cluster.
\section{Plug-and-Play Multi-Trajectory MPC}
To solve Problem \ref{pr:feedbackLaw}, the predicted agent's state trajectory has to be robust to possible network topology changes. 
Due to the time-varying nature of the constraints, a robust approach ensuring constraint satisfaction assuming a worst-case scenario in terms of new agents joining the network can lead to too conservative behaviour\cite{saccani2021autonomous}. 
To guarantee robustness against network changes and, at the same time, exploit the best of the current information about the network, we adopt the multi-trajectory MPC (mt-MPC) concept proposed in \cite{soloperto2022safe,saccani2022} and on the nonlinear tracking MPC controller proposed in \cite{limon2018nonlinear,fagiano2013generalized}.
The main idea, particularly suitable for time-varying constraints, consists in defining an MPC problem with two trajectories, sharing the first control action, in the same finite-horizon optimal control problem (FHOCP). The first is a safe trajectory towards a polytopic convex safe set $\hat{\pazocal{S}}_i(p_i(k))=\{q_i\in\mathbb{R}^3 : A_{c,i}(q_i-p_i(k))\leq b_{c,i}\}$ to guarantee the system's safety, here considered in the form of robustness against network changes. The second one, also called tracking trajectory, aims at minimizing a given tracking cost. 
Fig. \ref{fig:mt-MPC} shows a qualitative example where %three agents establish communication at two different time steps. T
the two trajectories for each agent can be easily distinguished as well as the safe sets. 
%At time $k$ (image on the left), the three agents design a safe trajectory (dashed line), that allows stopping the vehicle inside the current communication set. The designed safe trajectory remains a feasible solution at time $k+1$ (image on the right) despite the newly established communication.  
%The multi-trajectory MPC formulation allows thus to obtain a closed-loop behaviour that trades off safety and tracking of the reference.
We describe the approach for a single cluster %$\pazocal{G}_m(k)$ 
and the same problem is solved by each cluster $\pazocal{G}_m(k)$ with $m\in\{1,\dots,N_c(k) \}$. We denote with the superscripts ``t'', ``s'' the variables pertaining to the tracking and safe trajectory, respectively.
Furthermore, let us denote with $x_{i,(j|k)}^{(\cdot)}$ the predicted trajectory %tracking and safe state trajectories of the $i$-th subsystem~\eqref{eq:dt-system}
at time $k+j$ given the state at time $k$. 
Given a finite horizon $N\in \mathbb{N}$, we introduce the two tracking and safe input sequences $U_i^t=\left\{u_{i,(0|k)}^{t} \ u_{i,(1|k)}^{t} \ \dots \ u_{i,(N-1|k)}^{t}\right\}$, $U_i^s=\left\{u_{i,(0|k)}^{s} \ u_{i,(1|k)}^{s} \ \dots \ u_{i,(N-1|k)}^{s}\right\}$, 
%\begin{align}
%    U_i^t=\left\{u_{i,(0|k)}^{t} \ u_{i,(1|k)}^{t} \ \dots \ u_{i,(N-1|k)}^{t}\right\} \\
%    U_i^s=\left\{u_{i,(0|k)}^{s} \ u_{i,(1|k)}^{s} \ \dots \ u_{i,(N-1|k)}^{s}\right\},
%\end{align}
where $u_{i,(0|k)}^t=u_{i,(0|k)}^s$ is the first common control action.
 %Similarly, combining the local state and input constraints \eqref{eq:constraints}, the cluster constraints for system~\eqref{eq:dt-system-global} result in:
%\begin{align}
%    x\in \pazocal{X}:=\pazocal{X}_1\times\dots\times\pazocal{X}_{|\pazocal{V}_m(k)|},\\
%    u\in \pazocal{U}:=\pazocal{U}_1\times\dots\times\pazocal{U}_{|\pazocal{V}_m(k)|}.
%\end{align}
Now, given a collection of state references $r=$col$_{i\in\pazocal{V}_m(k)} (r_i)$, safe sets $\hat{\pazocal{S}}=\{\hat{\pazocal{S}}_i(p_i(k)), \forall i\in\pazocal{V}_m\}$ and positive scalars $\hat{J}^s(k)=$col$_{i\in\pazocal{V}_m(k)}(\hat{J}^s_i(k))$, whose derivation will be clarified in Section \ref{S:MPCingredients}, the following FHOCP $\mathscr{P}(x,r,\hat{\pazocal{S}},\hat{J}^s)$ is solved at each time step $k\geq 0$:
\begin{subequations}\label{eq:MPCproblem}
\begin{align}
      &\min  \limits_{U^{t,s}_i, \bar{x}^{s}_i, \bar{u}^{s}_i} \;\;
   	\sum^{|\pazocal{V}_m(k)|}_{i=1} J_i(x_i,U^{t,s}_i,\bar{x}_i^{s},&&r_i)\\
    &\text{subject to:} \nonumber\\  \ \ \ \ \
    &x^{t,s}_{i,(0|k)} = x_i(k), \\
    &u_{i,(0|k)}^t=u_{i,(0|k)}^s, \\
	% DYNAMICS
	&x_{i,(j+1|k)}^{t,s} = f_{i}(x_{i,(j|k)}^{t,s},u_{i,(j|k)}^{t,s}),&& \forall j \in \mathbb{N}_0^{N-1}\\
	&p^{t,s}_{i,(j|k)}=C_i x^{t,s}_{i,(j|k)}, && \forall j \in \mathbb{N}_0^{N}\\
	&\left(x_{i,(j|k)}^{t,s}, u_{i,(j|k)}^{t,s}\right) \in \pazocal{X}_i\times\pazocal{U}_i,  && \forall j \in \mathbb{N}_0^{N-1}\label{seq:stateconstr}\\
	&h_i\left( x_{\pazocal{N}_i(k),(j|k)}^{t,s}\right)\leq 0, &&\forall j \in \mathbb{N}_0^{N}  \label{seq:couplingconstr} \\
	&A_c \left(p_{i,(j+1|k)}^{s}-p_{i,(j|k)}^{s}\right) \leq \frac{b_c}{N}, &&\forall j \in \mathbb{N}_0^{N} \label{seq:safeSet}\\
	&x_{i,(N|k)}^{s} = \bar{x}^{s}_i = f_i(\bar{x}^{s}_i,\bar{u}^{s}_i),&&\label{seq:termConstr}\\
	& J_i^s(x_i,U_i^{s},\bar{x}_i^{s},r_i)\leq \hat{J}_i^s(k), \label{seq:convConstr}\\
	& \forall i \in \pazocal{V}_m(k). \nonumber
    \end{align}
\end{subequations}
The optimization variables $\bar{x}_i^{s}, \bar{u}_i^s$ are the artificial reference for the safe trajectory and~\eqref{seq:convConstr} is a convergence constraint and will be detailed in Section~\ref{SS:costandconvcostr}.
Constraint~\eqref{seq:safeSet}, instead, forces the positions of the predicted safe trajectory $p^s_{i,(j|k)}$ to lie inside a safe set $\hat{\pazocal{S}}(p_i(k))$, whose definition will be clarified in the next section. \\
Problem $\mathscr{P}(x,r,\hat{\pazocal{S}})$ is a nonlinear program (NLP), where the non-convex constraint~\eqref{seq:couplingconstr} is also a coupling constraint between neighbouring subsystems.
The MPC problem~\eqref{eq:MPCproblem} is amenable to distributed computation, and it can be solved with optimization algorithms for distributed non-convex optimization. Practical solutions to solve this problem are outside the scope of this work. 
Two possible solutions are presented: (i) the use of real-time iteration (RTI) with the Alternating Direction Method of Multipliers (ADMM) solver, as demonstrated in~\cite{carron2023multi}, where the single quadratic programming (QP) sub-problem can be solved in a distributed fashion, and convergence is obtained through RTI~\cite{gros2020linear}; or (ii) the approach proposed in~\cite{engelmann2020decomposition} under the assumption of fully communication within the cluster.
%Two possible solutions to address the problem are: i) follow the same line of the approach presented in~\cite{carron2023multi} using real-time iteration (RTI) where the solver is Alternating Direction Method of Multipliers (ADMM)~\cite{boyd2011distributed}, and thus the single quadratic programming (QP) sub-problem can be solved in a distributed fashion, and the convergence is obtained by RTI~\cite{gros2020linear}; ii) under fully communication assumption in the cluster, use the approach proposed in~\cite{engelmann2020decomposition}.
Thus, the MPC control law, computed by each cluster $\pazocal{G}_m$, can be locally computed by each vehicle and applied in a receding horizon fashion.
%%%%%%%%%%%%%%%%%%%%%%%%%%%%%%%%%%%%%%%%%%%%%
\section{MPC design and theoretical analysis} \label{S:MPCingredients}
In the following subsections we define and analyze the different elements defining problem~\eqref{eq:MPCproblem}, and conclude with a theoretical analysis of the MPC scheme.
\subsection{Cost function and convergence constraint} \label{SS:costandconvcostr}
%The proposed multi-trajectory MPC predicts two trajectories sharing the first common control action: the tracking trajectory aims at minimizing the tracking cost; the safe trajectory, similarly to a single-trajectory MPC, is used to guarantee the existence of a safe alternative in case of changes in the network. Since the approach is implemented in receding horizon, the common control action allows to obtain a less conservative behaviour of the system as shown in \cite{saccani2022}.
To maximize the benefit of the multi-trajectory approach \cite{saccani2022}, and partially decouple constraint satisfaction (safety) from cost function minimization (tracking), we would ideally minimize the cost only of the tracking trajectory and neglect the safe one. 
However, to guarantee the convergence of the MPC scheme, also the safe trajectory has to be included~\cite{soloperto2022safe}.
The local cost functions of problem~\eqref{eq:MPCproblem} is:
\begin{multline}
\label{eq:cost_mtMPC}
J_i(x_i,U_i^{t,s},\bar{x}_i^{s},r_i) =  \sum_{j=0}^{N-1}
l^t_i(x^t_{i,(j|k)}-\bar{x}_{r,i},u^t_{i,(j|k)}-\bar{u}_{r,i})\\
+\beta V_{O,i}^s(\bar{x}_i^s - r_i), 
\end{multline} 
where $l^{t}_i(\cdot,\cdot):\mathbb{R}^{n_i}\times \mathbb{R}^{m_i}\rightarrow\mathbb{R}$ is the stage cost function for the tracking trajectory,  
$V^{s}_{O,i}:\mathbb{R}^{d_i}\rightarrow\mathbb{R}$ is the offset safe cost function and  $\beta>0$ is a weight whose meaning will be better clarified later on. Note that for $\beta \rightarrow 0$, the cost function tends to the ideal case where only the tracking trajectory is considered in the cost~\cite{saccani2021autonomous}.
As shown in \cite{soloperto2022safe}, the mt-MPC formulation does not ensure convergence without including additional constraints to enforce a decrease in the safe cost function. % to the desired target.
To guarantee the convergence, as we will show in Section \ref{ss:theoretical analysis}, we have to properly design the constraint~\eqref{seq:convConstr}. 
To this aim let us define the following safe cost function representing a performance index for the safe trajectory:
\begin{align} \label{eq:safe_cost}
    J_i^s(x_i,U_i^{s},\bar{x}_i^{s},r_i) =  &\sum_{j=0}^{N-1}
l^s_i(x^s_{i,(j|k)}-\bar{x}_i^s,u^s_{i,(j|k)}-\bar{u}_i^s)\\
&+ V_{O,i}^s(\bar{x}_i^s - r_i), \nonumber
\end{align}
where $l^{s}_i(\cdot,\cdot):\mathbb{R}^{n_i}\times \mathbb{R}^{m_i}\rightarrow\mathbb{R}$ is a suitable safe stage cost. We can now define, at each time step $k$, an upper bound on the safe trajectory cost for the current time step, by exploiting the tail of the optimal safe trajectory computed at the previous time instant
\begin{align}
    \hat{J}^s_{i}(k)=&J_i^{s}(x_i(k-1),U^{s,*}_i(k-1),\bar{x}_i^{s,*},r_i) \\ &-l^s_i(x^{s,*}_{i,(0|k-1)}-\bar{x}_i^{s,*},u^{s,*}_{i,(0|k-1)}-\bar{u}_i^{s,*}), \nonumber
\end{align}
where the superscript ``*'' denotes the optimal quantities computed at time $k-1$ by solving the FHOCP~\eqref{eq:MPCproblem}. 
Thus, constraint~\eqref{seq:convConstr} imposes that the cost of the safe trajectory must be smaller or equal to the one obtained with the candidate safe trajectory computed at the previous time step.
Let us now define the feasibility set for the FHOCP~\eqref{eq:MPCproblem} as $\pazocal{F}\doteq \{ x:\mathscr{P}(x,r,\hat{\pazocal{S}},\hat{J}^s) \text{ admits a solution} \},$
and let assume that $\pazocal{F}$ is not empty and bounded.
Now, for all the elements $x_i(k)$ in $x(k)\in\pazocal{F}$ let us define the set of reachable steady state starting from $x_i(k)$ as follows.
\begin{defn}[Set of reachable steady states]
    \label{def:set_reachable}
    The set $\pazocal{R}_i$ contains all the steady states that can be reached by the system starting from the initial condition $x_i(k)$ in at most $N$ time steps with an admissible input sequence $V_i$. %while satisfying the position-dependent constraint $\hat{\pazocal{S}}_i(p_i(k))$, time-invariant state and input constraints $\pazocal{X}_i$, $\pazocal{U}_i$ convergence constraint~\eqref{seq:convConstr} and coupling obstacle avoidance constraint~\eqref{seq:couplingconstr}.
\begin{align}
&\pazocal{R}_i\left(x_i(k),x_{\pazocal{N}_{i(k)}},N,\hat{\pazocal{S}}_i(p_i(k),\hat{J}_i^s\right)\doteq \{  \bar{x}_i \in \pazocal{X}_i: \nonumber \\ 
    & \exists V_i\in\mathbb{R}^{m_i \times N}: \nonumber
    v_{i,(j|k)}\in\pazocal{U}_i, \forall j \in \mathbb{N}_0^{N-1}; \nonumber \\
    & x_{i,(0|k)}=x_i(k), \nonumber \\ 
    &x_{i,(j|k)}=f_i(x_{i,(j-1|k)},v_{i,(j-1|k)}), \forall j \in \mathbb{N}_0^N;\\
    &x_{i,(N|k)}=\bar{x}_i = f_i(\bar{x}_i,v_{i,(N-1|k)});\nonumber \\
    &x_{i,(j|k)}\in \pazocal{X}_i, \  C_i x_{i,(j|k)}\in \hat{\pazocal{S}}(p_i(k)), \forall j \in \mathbb{N}_0^N; \nonumber \\
    & h_i\left(x_{\pazocal{N}_{i(k),(j|k)}}\right)\leq0, \ \forall j \in \mathbb{N}_0^N; \nonumber \\
    & J_i^s(x_i,V_i^{s},\bar{x}_i,r_i)\leq \hat{J}_i^s \nonumber
\}.
\end{align}
\end{defn}
We are now in position to define the optimal reachable steady state belonging to $\pazocal{R}_i$.
\begin{defn}[Optimal reachable steady state] \label{def:opt_steady} The optimal reachable steady state $\bar{x}_i^o(k)$ is obtained at each time step by solving the following optimization problem:
\begin{align}
    \bar{x}_i^o(k)\in & \argmin \limits_{\substack{\bar{x}_i\in \pazocal{R}_i}}   V_{O,i}^s(\bar{x}_i-r_i).
\end{align}

\smallskip
\begin{figure*}[ht!]
	\centering
\setlength\abovecaptionskip{-0.3\baselineskip}
\setlength\belowcaptionskip{-1.1\baselineskip}
\includegraphics[width=0.9\textwidth]{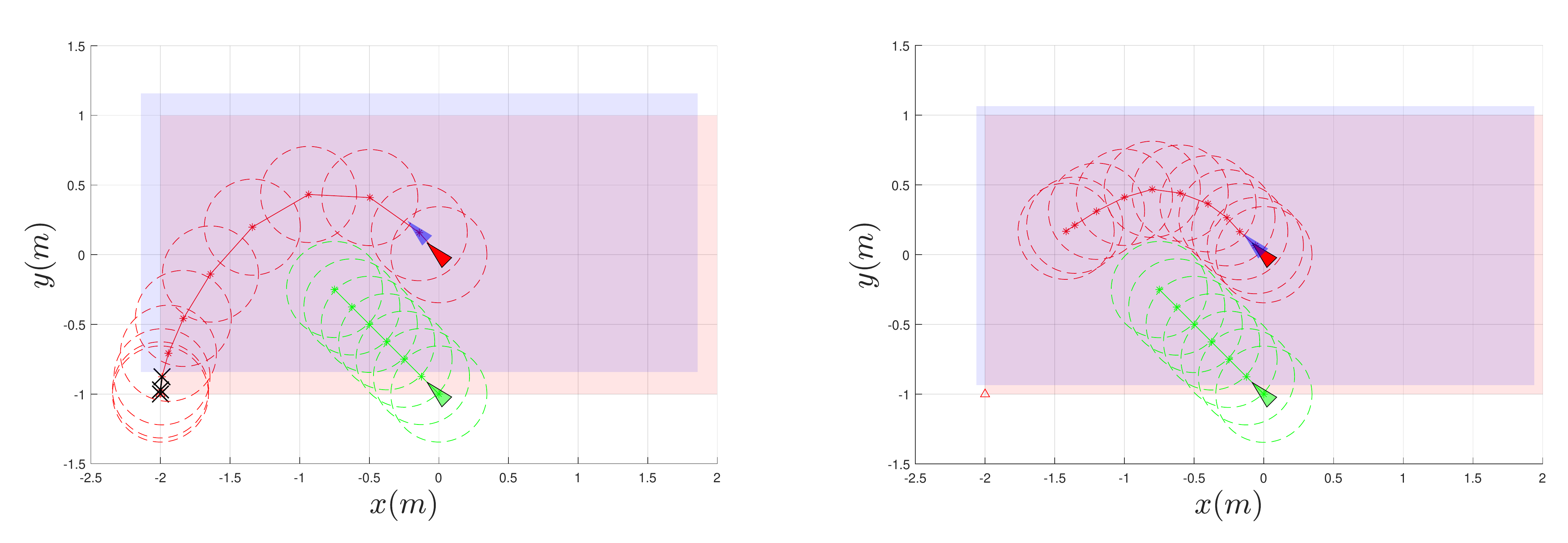}
\caption{Example showing the loss of recursive feasibility for a position-shifting state constraint by implementing~\eqref{eq:safeSetpoly} (left) and~\eqref{seq:safeSet} (right). We consider the trajectory of the red agent, while the green agent can be seen as a moving obstacle. Red line with `*' represents the predicted state trajectory at time $k$. Due to the presence of a terminal state constraint, the tail of this trajectory represents a candidate trajectory at time $k+1$ (see~\eqref{eq:candidateSol}). Light red polytope represents the state constraint at time $k$, while light blue polytope represents the state constraint shifted at the first predicted state $x_{1|k}$. In both examples, an obstacle avoidance constraint is imposed, and the dashed circles around the trajectories represent the size of the agents.}
	\label{fig:feas}
\end{figure*}
%For the optimal reachable steady state we can compute the optimal safe cost $J_i^{s,o}(x_i,U^{s,o}_i,\bar{x}_i^{s,o},r_i)$.
\end{defn}
%\begin{multline*} J_i^{s,o}(x_i,U^{s,o}_i,\bar{x}_i^{s,o},r_i)=\sum_{i=0}^{N-1} l^s_i (x^{s}_{i(j|k)}-\bar{x}^{s,o}_i(k),\\u^{s}_{i,(j|k)}-\bar{u}^{s,o}_i(k)) +V_{O,i}^s(\bar{x}^{s,o}_i(k) -r_i). \end{multline*}
Thus, two possibilities may arise during the navigation: \\
i) There exists a time instant $\bar{k} \geq 0$, such that the final target~$r_i$ becomes reachable for a long enough horizon $N$, i.e.
$
        \exists \bar{k} : \bar{x}^o_i(k)=r_i, \ \bar{x}^o_i(k)\in\pazocal{R}_i,  \  \forall k\geq \bar{k}. %\left(x_i(k),x_{\pazocal{N}_{i(k)}},N,\hat{\pazocal{S}}_i(p_i(k))\right),
$ \\
ii) There exists a time instant $\bar{k} \geq 0$, such that the final target $r_i$ cannot be reached, e.g. because an agent is between the vehicle and the target, but the position associated to the optimal steady state remains constant, $\forall k \geq \bar{k}$, i.e.\\
$ \exists \bar{k} :      \bar{x}^o_i(k)=\bar{x}^o_i\in\pazocal{R}_i, \        %\left(x_i(k),x_{\pazocal{N}_{i(k)}},N,\hat{\pazocal{S}}_i(p_i(k))\right),\\
         C_i \bar{x}_i^o \in \hat{\pazocal{S}}(p_i(k)), \  \forall k\geq \bar{k}.$ \\
Finally, similarly to \cite{limon2018nonlinear,fagiano2013generalized}, let us consider the following assumptions on the cost functions.
\begin{assum}\label{ass:cost_terminal} 
$f$ and $l^{(\cdot)}$ are continuous on $\bar{\pazocal{F}} \times \pazocal{U}$, where $\bar{\pazocal{F}}$ is the closure of $\pazocal{F}$, hence $\exists \alpha_f, \alpha_l\in\pazocal{K}_{\infty}: \|f_i(\bar{x},\bar{u})-f_i(\hat{x},\hat{u})\| \leq \alpha_f (\| (\bar{x},\bar{u})-(\hat{x},\hat{u}) \|)$, $|l_i^{(\cdot)}(\bar{x},\bar{u})-l_i^{(\cdot)}(\hat{x},\hat{u})| \leq \alpha_l (\| (\bar{x},\bar{u})-(\hat{x},\hat{u}) \|)$, $\forall (\bar{x},\bar{u}),(\hat{x},\hat{u})\in \pazocal{F}\times\pazocal{U}$, $\forall i$, for some vector norm $\| \cdot \|$. The stage cost $l_i^{(\cdot)}(x,u)$ is designed such that $l_i^{(\cdot)}(x,u)\geq \alpha_l (|x|)$, where $\alpha_l$ is a $\pazocal{K}_{\infty}$ function and the safe offset cost $V_{O,i}^{s}(\bar{x}_i-r_i)$ is positive definite, strictly convex and subdifferentiable functions with a unique minimizer that is
$\bar{x}^o_i=\arg\min\limits_{\bar{x}_i}V_{O,i}^{s}(\bar{x}_i-r_i).$ 
%Moreover, there exists a $\pazocal{K}_{\infty}$ function $\alpha_O$ such that
%    \begin{equation*}
 %       V_{O,i}^{s}( \bar{x}_i-r_i)-V_{O,i}^{s}(\hat{x}_i-r_i)\geq\alpha_O(\|\bar{x}_i-\hat{x}_i\|).
%    \end{equation*}
\end{assum}
Where a continuous, strictly increasing function $\alpha:[0,+\infty)\rightarrow[0,+\infty)$ is said to belong to class $\pazocal{K}_{\infty}$ if $\alpha(0)=0$ and if $\lim \limits_{a\rightarrow+\infty} \alpha(a) =+\infty $.
\subsection{Collision avoidance} \label{ss:collisionavoid}
%To avoid collisions between vehicles, a coupling constraint between neighboring systems is considered. 
By defining with $\pazocal{Q}(x_i(k))$ the space occupied by the vehicle $i$ at time $k$ we want to guarantee:
\begin{equation} \label{eq:oa_set}
    \pazocal{Q}(x_i(k))\cap \pazocal{Q}(x_j(k)) = \emptyset, \ \ \forall j \in \pazocal{N}_i(k) \backslash \{ i\},
\end{equation}
and thus we can rewrite constraint~\eqref{eq:oa_set}, $\forall j \in \pazocal{N}_i(k)\backslash \{ i\}$ as:
$$h_i(x_{\pazocal{N}_i(k)})= -\text{dist}(\pazocal{Q}(x_i),\pazocal{Q}(x_j))+d_{min} \leq0 $$
where $d_{min}\geq 0$ is a desired safety margin and the distance $\text{dist}(\pazocal{Q}(x_i),\pazocal{Q}(x_j))$ is defined as $\text{dist}(\pazocal{Q}(x_i),\pazocal{Q}(x_j)):= \min \limits_{t} \{ \| t \| : (\pazocal{Q}(x_i)+t) \cap \pazocal{Q}(x_j) \neq \emptyset \}$.
In \cite{zhang2020optimization}, it is shown how to rewrite~\eqref{eq:oa_set} as a smooth, differentiable constraint when a polytopic or ellipsoidal shape is considered. 
Let us consider the polytopic shape $\pazocal{T}(p_i)=\{p_i \in \mathbb{R}^{3}: A_i p_i \leq b_i \}$ (see \cite{zhang2020optimization} for extension to other convex sets). In this case, the space occupied by the vehicle~$i$ can be defined as the translation and rotation of the initial polytopic set $\pazocal{T}(p_i)$.
\begin{equation}
    \pazocal{Q}(x_i)=R(x_i)\pazocal{T}(p_i)+t(x_i), \ \ \ 
\end{equation}
where $R:\mathbb{R}^{n_i}\rightarrow \mathbb{R}^{3\times 3}$ is an (orthogonal) rotation matrix (see~\cite{zhang2020optimization}) and $t:\mathbb{R}^{n_i}\rightarrow \mathbb{R}^{3}$ is the translation vector.
To avoid collision between vehicle $i$ and neighbor $j$ we can write the constraint $h_i(x_{\pazocal{N}_{i}(k)})$:
\begin{align} \label{eq:OA_set}
    \text{dist}(\pazocal{Q}(x_i)&,\pazocal{Q}(x_j))>d_{min} \\
    & \Longleftrightarrow \exists \lambda \geq 0, \ \mu \geq 0: \nonumber \\
    & -(b_i-A_i t(x_j))^T \lambda -(b_j-A_j t(x_i))^T \mu > d_{min}, \nonumber\\
    & R(x_j)^T A_i^T \lambda + R(x_i)^T A_j^T \mu = 0, \ \ \| A_j^T \mu \| \leq 1, \nonumber
\end{align}
\begin{rem}
To the benefit of a reduced conservatism in the approximation of the vehicle's shape, constraint~\eqref{eq:OA_set} increases the complexity of the optimization problem due to the need for additional optimization variables. %with respect to~\eqref{eq:oa_circular}
Alternatively, the complexity can be reduced by over-approximating the shape of the vehicles with a sphere $\pazocal{T}_S(p_i)=\{ p_i \in \mathbb{R}^3: \ \|p_i\|_2 \leq \sigma \}$.
Thus, it can be imposed that the euclidean distance between the vehicle $i$ and $j$ is greater or equal than twice the radius $\sigma$ accounting for the maximum size of the vehicle $\forall j\in \pazocal{N}_i(k)\backslash \{ i\}$:
\begin{equation} \label{eq:oa_circular}
h_i(x_{\pazocal{N}_i(k)}) = -
    \| p_i - p_j \|^2_2 + (2\sigma)^2+d_{min}\leq0.
\end{equation}
\end{rem}
\subsection{Safe set}
To avoid collisions among different agents, it is crucial to be able to reach a steady state within the communication set~$\pazocal{C}_i(p_i(k))$.
%This is depicted in Fig. \ref{fig:mt-MPC}, where the safe (dashed) trajectory remains a feasible solution despite the newly established communication.
To define the safe set~$\pazocal{\hat{S}}_i(p_i)$ considered in constraint~\eqref{seq:safeSet}, where the safe trajectory can remain, let us firstly consider the size of the vehicle, by means of the set
\begin{equation}
    \pazocal{O}=\{ R(x_i)\pazocal{T}(p_i), \  \forall x_i \in (\pazocal{X}_i\times\pazocal{U}_i)^{\perp}  \},
\end{equation}
where $(\pazocal{X}_i\times\pazocal{U}_i)^{\perp}$ is the projection of $\pazocal{X}_i\times\pazocal{U}_i$ on the state space.
Thus, we use the computed set to tighten the communication set~$\pazocal{C}_i(p_i)$
\begin{equation} \label{eq:safeSetspeed}
    \pazocal{S}_i(p_i)=\pazocal{C}_i(p_i) \ominus \pazocal{O},
\end{equation}
where $\ominus$ is the Pontryagin difference.
\begin{rem}
Note that by considering constraint~\eqref{eq:oa_set}, the difference with the set $\pazocal{O}$ can be avoided at the cost of checking at each time step the belonging of $\pazocal{T}(p_i)$ to the set $\pazocal{C}_i(p_i)$. This can be done by adding additional optimization variables to check the inclusion of a convex polytope, as shown for example in \cite{blanchini1999set}.
\end{rem}
The obtained set $\pazocal{S}_i(p_i)$ represents a safe region where the vehicle can safely counteract to a possible change in the topology and the set $\pazocal{O}$ accounts for the size of the vehicle.
Finally, we under-approximate the set $\pazocal{S}_i(p_i)$ with the following polytopic set 
\begin{equation} \label{eq:safeSetpoly}
    \hat{\pazocal{S}}_i(p_i) \doteq \{p_i,q\in \mathbb{R}^3 : A_c (q-p_i) \leq b_c \},
\end{equation}
that, in practice, can be easily computed by performing the convex hull of some samples of the borders of~$\pazocal{S}_i(p_i)$.
Constraint~\eqref{eq:safeSetpoly}, is always centered at the vehicle's position~$p_i(k)$.
This constraint's feature, however, can cause a loss of feasibility. 
To guarantee that the problem is recursively feasible, we need to ensure that the tail of the predicted safe trajectory lies in the safe sets generated by shifting the set on the predicted trajectory $p^s_{i(j|k)}\in\hat{\pazocal{S}_i}(p^s_{i(l|k)})$, $\forall l \in \mathbb{N}_0^j$.
Fig.~\ref{fig:feas} shows an example, where the candidate trajectory of the red agent computed at time $k$ is not feasible for the constraint at the subsequent time step $k+1$.
To address this issue, we propose the implementation~\eqref{seq:safeSet} to satisfy constraint~\eqref{eq:safeSetpoly}. %for constraint
%\begin{equation} \label{eq:recFeasSafeSet}
%    	A_c \left(p_{i,(j+1|k)}^{s}-p_{i,(j|k)}^{s}\right) \leq \frac{b_c}{N}.
%\end{equation}
Specifically, constraint~\eqref{seq:safeSet} imposes that the predicted trajectory is able to reach the border of the polytopic safe set~\eqref{eq:safeSetpoly} only in $N$ time steps.
%%%%%%%%%%%%%%%%%%%%%%%%%%%%%%%%%%%%%%%%%%%%%%%%%
\subsection{Theoretical analysis} \label{ss:theoretical analysis}
We now analyze the theoretical guarantees of the proposed MPC scheme, by exploiting the different ingredients described in the previous section.
Before stating the main result, let us consider the following assumption.
\begin{assum}\label{ass:solFHOCP}
For any $x\in\pazocal{F}$, there exists a finite value of the optimal safe cost $J^{s,o}$ and the FHOCP $\mathscr{P}(x,r,\hat{\pazocal{S}},\hat{J}^s)$ has at least one global minimum, which is computed by the solver independently from its initialization.
\end{assum}
The latter assumption is quite usual and implicitly considered in the context of economic MPC and nonlinear MPC~\cite{fagiano2013generalized}. Moreover, it is satisfied if the FHOCP is convex, which is the important case of linear systems with convex constraints and convex tracking stage cost.\\
We are now in position to state the following proposition:
\begin{prop}\label{prop:convandrec}
Let Assumptions [\ref{ass:posInv}-\ref{ass:solFHOCP}] be satisfied and assume that the FHOCP~\eqref{eq:MPCproblem} at time $k=0$ is feasible.
Then, problem~\eqref{eq:MPCproblem} is recursively feasible and the system controlled by the MPC controller derived from its solution converges arbitrarily close to the optimal admissible equilibrium point according to Definition~\ref{def:opt_steady} while satisfying constraints $\forall k\geq0$.
\end{prop}  
The Appendix contains the proof of Proposition~\ref{pr:feedbackLaw}.

%%%%%%%%%%%%%%%%%%%%%%%%%%%%%%%%%%%%%%%%%%%
\section{Numerical example}
In this section, we show the effectiveness of the approach via numerical simulations.
%by considering two scenarios presented in the following subsections.
%Firstly, we present a comparison with a standard MPC formulation, showing that a relatively simple example can lead to a loss of feasibility. Then, we propose a scenario where multiple autonomous vehicles have to cross a traffic intersection.
We employed a Quad-Core Intel Core i7 (2.8 GHz, 16 GB) on MATLAB 2020b under MS Windows, using CasADI \cite{Andersson2019} to build problem~\eqref{eq:MPCproblem} %, ALADIN-$\alpha$ \cite{engelmann2022aladin} to distribute the optimization problem 
and IPOPT \cite{wachter2006implementation} to solve local non-convex problems.
\begin{figure}
    \centering    
    \setlength\abovecaptionskip{-0.2\baselineskip}
    \setlength\belowcaptionskip{-1\baselineskip}
    \includegraphics[width=.78\columnwidth]{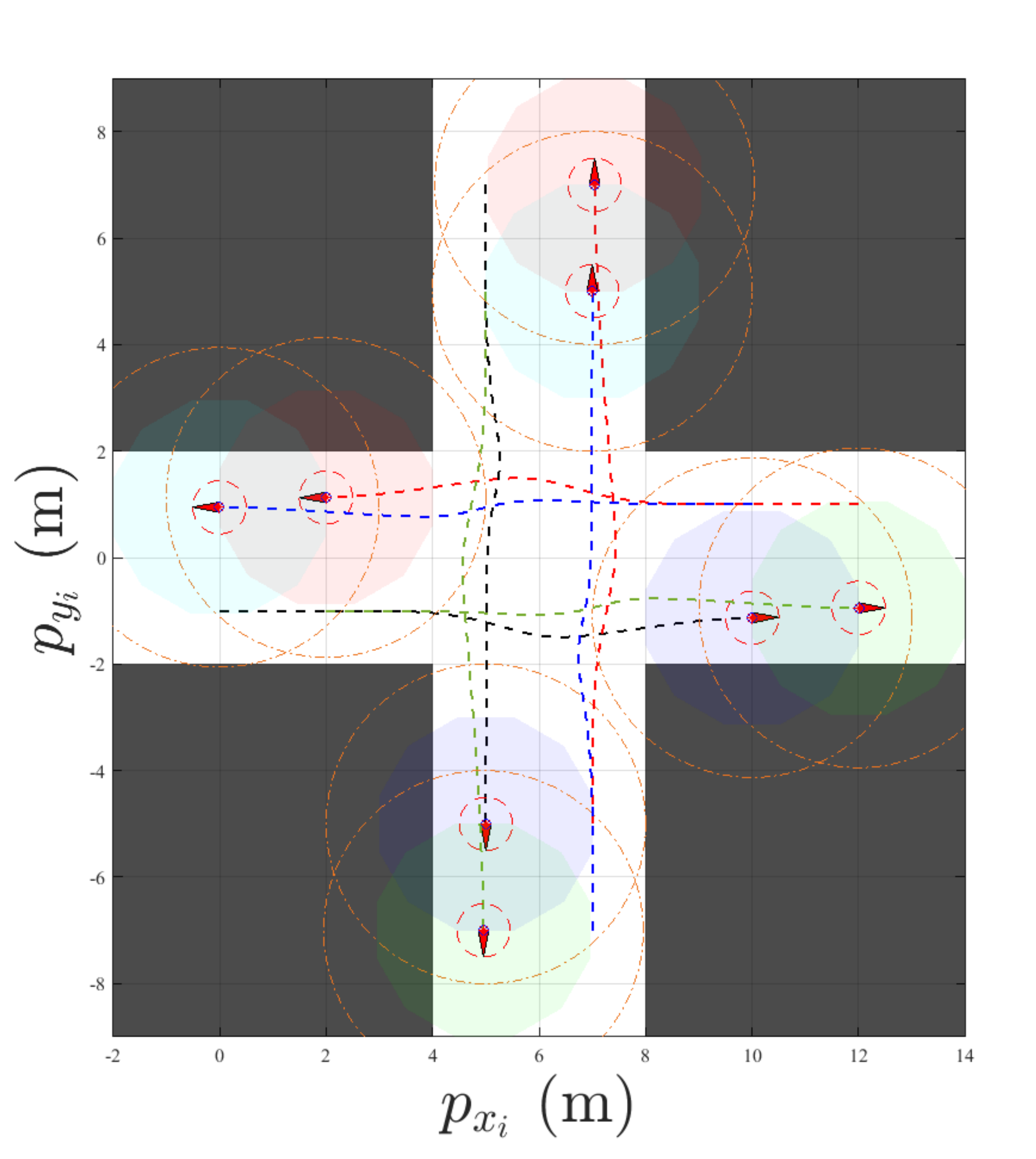}
    \caption{Simulation results of the intersection problem considered. Dashed colored lines represent vehicle trajectories. Red triangles are the final pose of the vehicles. Orange dashed circles represent the communication set and colored polytopes the safe sets $\hat{\pazocal{S}}_i(p_i)$.}
    \label{fig:example}
\end{figure}
% \begin{figure*}[ht] 
% \centering
% \subfloat[][\label{subfig:incrocio}]
% 		{\includegraphics[width=.42\textwidth]{Figures/Incrocio.png}}\quad
% 		%
% \subfloat[][\label{subfig:simRes}]
% 		{\includegraphics[width=.40\textwidth]{Figures/crossroad.pdf}}\quad
% \caption{Illustration of the intersection problem considered (\ref{subfig:incrocio}). In (\ref{subfig:simRes}) simulation results are shown. Dashed colored lines represent vehicle trajectories. Red triangles are the final pose of the vehicles. Orange dashed circles represent the communication set and colored polytopes the safe sets $\hat{\pazocal{S}}_i(p_i)$.}
% \label{fig:example}
% \end{figure*}
%In this simulation, %shown in Fig.~\ref{fig:example}, 
We consider different vehicles approaching an intersection representing, for example, mobile robots in a warehouse.
We consider eight ground vehicles described by the following discrete time kinematic model:
\begin{equation*}\label{eq:dt-car_model}
    x_i(k+1)= \left(
    \begin{matrix}
    p_{x_i}(k)+T_s \cos{(\theta_i)}v_i(k)
    \\
    p_{y_i}(k)+T_s \sin{(\theta_i)}v_i(k)
    \\
    v_i(k)+T_s a_i(k)    \\
    \theta_i(k) + T_s \frac{v_i(k)}{L} \tan{(\gamma_i(k))}
    \\
    \gamma_i(k)+T_s \delta_i(k)
    \end{matrix} \right),
\end{equation*}
where $T_s$ is the sampling time, $\begin{bmatrix} p_{x_i} \\ p_{y_i}\end{bmatrix}\in \mathbb{R}^2$ is the vehicle's position and $\theta_i$, $\gamma_i$ $\in \mathbb{R}$ are the yaw and the steering angles. The inputs are the commanded acceleration $a_i\in \mathbb{R}$ and the steering rate $\delta_i\in\mathbb{R}$.
Thus, each local system has a state $x_i\in\mathbb{R}^5$ and an input $u_i \in\mathbb{R}^2$. The sampling time $T_s$ is $0.1$ $s$, and the wheelbase of the vehicle $L$ is $0.8$ $m$.
We assume a circular communication area $\pazocal{C}_i(p_i)=\{ p_i\in\mathbb{R}^2: \|p_i\|_2 \leq 3 m \}$ and we assume a circular shape for the vehicles with a diameter of $1$ $m$, leading to the non convex coupling constraint~\eqref{eq:oa_circular} and a set $\pazocal{O}=\{ p_i \in \mathbb{R}^2: \ \| p_i \|_2 \leq \sigma \}$ where $\sigma = 1$ $m$.
The acceleration and the steering rate are limited to $|a_i|\leq3$ $m$/$s^2$ and $|\delta_i|\leq 1$ rad/s. The velocity and the steering angle are limited to $v_i\leq2$ $m$/$s$ and $\gamma_i \in [-\frac{\pi}{2}, \ \frac{\pi}{2}]$ $rad$ and the yaw $\theta_i \in [0, \ 2\pi]$ $rad$. 
Each vehicle has to cross the workspace by guaranteeing the obstacle avoidance constraint and solving a problem only with the neighbouring subsystems.
Fig. \ref{fig:example} shows the trajectories obtained in closed loop together with the communication areas $\pazocal{C}_i(p_i)$ and the safe sets $\hat{\pazocal{S}}_i(p_i)$.
Finally Fig. \ref{fig:graph} shows the evolution of the time-varying communication network with plug and play operations at different iterations.
\begin{figure*}[h!]
	\centering
     \setlength\abovecaptionskip{-0.2\baselineskip}
\setlength\belowcaptionskip{-1.1\baselineskip}
\includegraphics[width=.9\textwidth]{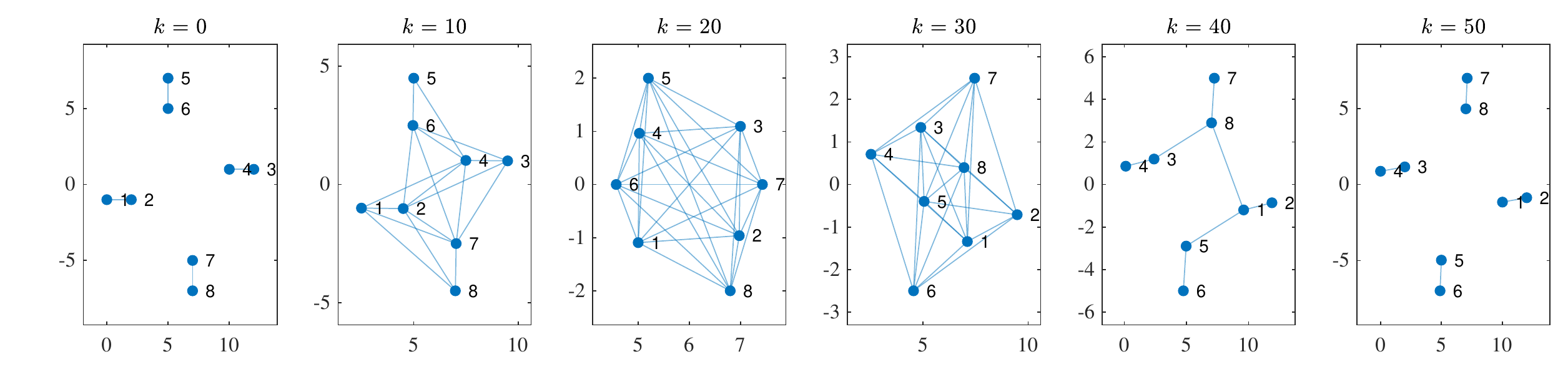}
	\caption{Evolution of the communication graph at iteration $k=0,10,20,30,40,50$}
	\label{fig:graph}
\end{figure*}
%\begin{figure}[h!]
%	\centering	\includegraphics[width=.85\columnwidth]{Figures/graph_2.pdf}	\caption{Evolution of the communication graph at iteration $k=0,10,20,30,40,50$}	\label{fig:graph}
%\end{figure}
%%%%%%%%%%%%%%%%%%%%%%%%%%%%%%%%%%%%%%%
\section{Conclusion}
In this work, we propose a multi-trajectory MPC for the trajectory generation of multi-agent systems able to handle changes in the communication network topology in real-time without request.
We proved that the approach, based on nonlinear tracking MPC, makes the agents safely converge to their reference or to the closest admissible steady state.
Finally, a numerical example demonstrates the effectiveness of the approach in a traffic intersection.
The current research activities aim to test the approach's effectiveness on a custom hardware platform with experimental evaluation.

\appendix
%\section*{Proof of proposition 1}
Before proceeding with the proof of Proposition~\ref{prop:convandrec}, the following lemma must be introduced.
\begin{lemma}
Let Assumption [\ref{ass:posInv}-\ref{ass:solFHOCP}] hold. Thus, for any $\epsilon>0$, there exists a finite value $\underline{\beta}(\epsilon)$ such that, if $\beta\geq\underline{\beta}(\epsilon)$ in~\eqref{eq:cost_mtMPC}, $\forall x\in\pazocal{F}$ and any $\hat{J}^s\geq J^{s,o}$ then
\begin{equation} \label{eq:optimal_offset}
    V_O^s(\bar{x}^{s,*}-r)-V_O^{s}(\bar{x}^o-r)\leq  \epsilon,
\end{equation}
where $\bar{x}^{s,*}$ is the terminal safe state computed by solving $\mathscr{P}(x,r,\hat{\pazocal{S}},\hat{J}^s)$ and $V_O^s = \sum_{i=1}^{|\pazocal{V}_m|} V_{O,i}^s$.
\end{lemma}
\begin{proof}
Let $\tilde{\bar{x}}_i$ be a steady state such that: $
%\begin{align*}
    \tilde{\bar{x}}_i\in \argmin \limits_{\bar{x}_i\in \pazocal{R}_i}   V_{O,i}^s(\bar{x}_i-r_i)
%\end{align*}
$ and $\tilde{\bar{u}}_i$ the associated control input such that $\bar{x}_i=f_i(\bar{x}_i,\bar{u}_i)$.
Let $\tilde{V}_i^s\in \pazocal{U}_i$ be a sequence of $N$ control actions such that $x^s_{i,(N|k)}=\tilde{\bar{x}}_i$ and $\tilde{v}^s_{i,(N-1|k)}=\tilde{\bar{u}}_i$. This sequence exists according to Definition~\ref{def:set_reachable}. 
%Furthermore, since $\tilde{\bar{x}}_i$ is the unique minimizer of $V_{O,i}^s$, we have that $J_i^s(x_i,\tilde{V}_i^{s},\bar{x}^{s}_i,r_i)=J_i^{s,o}(x_i,U^{s,o}_i,\bar{x}_i^{s,o},r_i)\leq \hat{J}^s_i$. 
Let also consider a sequence $\tilde{V}^t_i\in \pazocal{U}_i$ of~$N$ control actions such that: $\tilde{v}^t_{i,(0|k)}=\tilde{v}^s_{i,(0|k)}$. Thus, the sequences $\tilde{V}^{t,s}=$col$_{i\in\pazocal{V}_m(k)}(\tilde{V}^{t,s}_i)$, %$\tilde{V}^t=$col$_{i\in\pazocal{V}_m(k)}(\tilde{V}^t_i)$ 
and the steady state $\tilde{\bar{x}}=$col$_{i\in\pazocal{V}_m(k)}(\tilde{\bar{x}}_i)$, $\tilde{\bar{u}}=$col$_{i\in\pazocal{V}_m(k)}(\tilde{\bar{u}}_i)$ are a feasible solution for problem $\mathscr{P}(x,r,\hat{\pazocal{S}},\hat{J}^s)$.
The cost associated with $\tilde{V}^{t,s}$ is
$
    J(x,\tilde{V}^{t,s},\tilde{\bar{x}},r) = \sum_{i=1}^{|\pazocal{V}_m|}\sum_{j=1}^{N-1} l^t_i(\tilde{\xi}^t_{i,(j|k)}-\tilde{\bar{x}}_i,\tilde{v}^t_{i,(j|k)}-\tilde{\bar{u}}_i) + \beta V_{O,i}^s(\tilde{\bar{x}}_i-r_i),
$
where $\tilde{\xi}^t_{i,(j|k)}$, $\forall j\in\mathbb{N}_0^N$ is the state trajectory obtained applying the sequence $\tilde{V}^t_i$.\\
Let us now consider, any other possible state $\hat{\bar{x}}_i\in\pazocal{R}_i$ 
%\left(x_i(k),x_{\pazocal{N}_{i(k)}},N,\hat{\pazocal{S}}_i(p_i(k))\right)$ 
and input $\hat{\bar{u}}_i\in\pazocal{U}_i$ such that $\hat{\bar{x}}_i=f_i(\hat{\bar{x}}_i,\hat{\bar{u}}_i)$, $\forall i \in \pazocal{V}_m$ and the input sequences $\hat{V}^{t,s}$ 
%Thanks to Definition~\ref{def:set_reachable}, let us consider the two input sequences $\hat{V}^s$, $\hat{V}^t$ 
such that $V_{O}^s(\hat{\bar{x}}-r)- V_{O}^s(\tilde{\bar{x}}-r)>\epsilon$. Thus, $\hat{V}^{t,s}$, $\hat{\bar{x}}$, $\hat{\bar{u}}$ are a feasible solution for problem 
$\mathscr{P}(x,r,\hat{\pazocal{S}},\hat{J}^s)$.
The cost associated with $\hat{V}^{t,s}$ is
$
    J(x,\hat{V}^{t,s},\hat{\bar{x}},r) = \sum_{i=1}^{|\pazocal{V}_m|}\sum_{j=1}^{N-1} l^t_i(\hat{\xi}^t_{i,(j|k)}-\hat{\bar{x}}_i,\hat{v}^t_{i,(j|k)}-\hat{\bar{u}}_i)+  \beta V_{O,i}^s(\hat{\bar{x}}_i-r_i).
$
Thus, the difference between $J(x,\tilde{V}^{t,s},\tilde{\bar{x}},r)$ and $J(x,\hat{V}^{t,s}(k),\hat{\bar{x}},r)$, denoted for simplicity from now on as $J(x,\tilde{V}^{t,s})$ and $J(x,\hat{V}^{t,s})$ is given by:
\begin{multline}
    J(x,\tilde{V}^{t,s})- J(x,\hat{V}^{t,s}) = \beta[V_{O}^s(\tilde{\bar{x}}-r) - V_{O}^s(\hat{\bar{x}}-r)] \\
    +\sum_{i=1}^{|\pazocal{V}_m|}\sum_{j=1}^{N-1} [ l^t_i(\tilde{\xi}^t_{i,(j|k)}-\tilde{\bar{x}},\tilde{v}^t_{i,(j|k)}-\tilde{\bar{u}}) \\ -l^t_i(\hat{\xi}^t_{i,(j|k)}-\hat{\bar{x}},\hat{v}^t_{i,(j|k)}-\hat{\bar{u}})].
\end{multline}
By exploiting Assumption \ref{ass:cost_terminal}, we obtain: $
    J(x,\tilde{V}^{t,s})- J(x,\hat{V}^{t,s}) < - \beta\epsilon +\eta,
$ where (see Appendix of \cite{fagiano2013generalized} for the complete derivation)
$%\begin{equation}
    \eta = \sum_{i=0}^{N-1} \sum_{j=0}^{i} \alpha_l \left( \alpha_f^{(i-j)} \left( \max \limits_{\tilde{v},\hat{v}\in\pazocal{U}}\| \tilde{v}-\hat{v} \| \right) \right)>0,
%\end{equation}
$
and $\alpha_f^{(i)}(a)\dot{=} \underbrace{\alpha_f (\alpha_f(\dots \alpha_f \dots ))}_{\text{$i$ times}}$, $\alpha_f^{(0)}(a)\dot{=}a$, $\pazocal{U}:=\pazocal{U}_1\times\dots\times\pazocal{U}_{|\pazocal{V}_m(k)|}$.
Thus by selecting a value of $\beta(\epsilon) = \eta/\epsilon$ and $\beta\geq\underline{\beta}(\epsilon)$, we obtain:
\begin{equation} \label{eq:CH3_ineq_cost}
    J(x,\tilde{V}^{t,s}) < J(x,\hat{V}^{t,s}),
\end{equation}
Now, let us suppose by contraddiction that the solution $U^{t,s,*}$ to the FHOCP is such that $V_{O}^s(\bar{x}^{s,*}-r) - V_{O}^s(\bar{x}^o-r)\geq\epsilon$. Then equation \eqref{eq:CH3_ineq_cost} is still valid when $J(x,\hat{V}^{t,s})=J(x,U^{t,s,*})$. However, this cannot be possible, since by Assumption \ref{ass:solFHOCP}, the solver computes the global minimizer of the FHOCP.
\end{proof}
\noindent We are now in position to introduce the following proposition presenting the convergence properties of the mt-MPC approach:
\begin{prop} \label{lem:CH3_convergence}
Suppose that Assumptions [\ref{ass:posInv}-\ref{ass:solFHOCP}] hold and consider a given state setpoint $r=f(\bar{x}_r,\bar{u}_r)$. Let us select a value of $\epsilon>0$ and $\beta\geq \underline{\beta}(\epsilon)$.
Then for any feasible initial state $x(0)$, the system $f(x(k),u(k))$ controlled by the mt-MPC law $\kappa(x,r)$ obtained by solving the FHOCP \eqref{eq:MPCproblem} satisfies the constraints, is stable, and converges to a steady state such that~\eqref{eq:optimal_offset} holds. 
\end{prop}
\begin{proof}
The proof is divided in two parts. Firstly we prove that the problem is recursively feasible.
Then, we prove the asymptotic stability
of the equilibrium point $(\bar{x}^{s,*},\bar{u}^{s,*})$. 
\\
\textit{Recursive feasibility}: We denote the solution of problem~\eqref{eq:MPCproblem} at time $k$ as $U_{i,k}^{t,s*}=[ u^{t,s*}_{i,(0|k)}, \ u^{t,s*}_{i,(1|k)}, \dots , u^{t,s*}_{i,(N-1|k)}]$, the corresponding optimal predicted state trajectories $X_{i,k}^{t,s*}=[ x^{t,s*}_{i,(0|k)}, \ x^{t,s*}_{i,(1|k)}, \dots, x^{t,s*}_{i,(N|k)}]$, the optimal safe artificial reference $(\bar{x}_{i,k}^{s*}$, $\bar{u}_{i,k}^{s*})$ and with $\pazocal{C}_{i,k}$, $\hat{\pazocal{S}}_{i,k}$ the sets $\pazocal{C}_{i}(p_i(k))$ and $\hat{\pazocal{S}}_{i}(p_i(k))$ at time $k$ respectively. 
As standard in MPC, let us define a candidate solution at time $k+1$.
Due to the possible variation of $\pazocal{N}_{i}(k+1)$, we cannot guarantee that the optimal tracking trajectory at time $k$ can be used to compute a candidate solution at time $k+1$.
Let us firstly observe from~\eqref{eq:MPCproblem}, that the safe trajectory is a sub-optimal solution for the tracking trajectory since it is subject to the same constraints plus constraint~\eqref{seq:safeSet} and~\eqref{seq:convConstr}.
Let us define the candidate solution based on the safe trajectory as follows:
\begin{align}\label{eq:candidateSol}
%\hat{\bar{x}}^{s}_{i,k+1}&=\bar{x}^{s*}_{i,k} \ \ \ \ \ \hat{\bar{u}}^{s}_{i,k+1}=\bar{u}^{s*}_{i,k}  \nonumber \\
    \hat{U}_{i,k+1}^{t,s}&=[ u^{s*}_{i,(1|k)}, \ \dots \ , u^{s*}_{i,(N-1|k)}, \  \bar{u}^{s*}_{i,k}]  \\
    \hat{X}_{i,k+1}^{t,s}&=[ x^{s*}_{i,(1|k)}, \ \dots \ , x^{s*}_{i,(N|k)}, \ \bar{x}^{s*}_{i,k}], \nonumber
\end{align}
and the artificial reference is $\hat{\bar{x}}^{s}_{i,k+1}=\bar{x}^{s*}_{i,k} $,  $\hat{\bar{u}}^{s}_{i,k+1}=\bar{u}^{s*}_{i,k}$.
In the following we analyze how the candidate trajectories~\eqref{eq:candidateSol} satisfy constraints~\eqref{seq:couplingconstr}-\eqref{seq:safeSet} and~\eqref{seq:convConstr}, while the satisfaction of other constraints in~\eqref{eq:MPCproblem} is straightforward and will be omitted.
Satisfaction of~\eqref{eq:safeSetpoly}, implemented as~\eqref{seq:safeSet} is trivial thanks to its construction.
At time $k+1$ for the trajectory~\eqref{eq:candidateSol} we have $
    A_c\sum_{j=2}^{N}\left(p_{i,(j|k)}^{s}-p_{i,(j-1|k)}^{s}\right)+A_c \left(\bar{p}_{i,k}^{s*}-p_{i,(N|k)}^{s}\right)\leq
    A_c\sum_{j=2}^{N}\left(p_{i,(j|k)}^{s}-p_{i,(j-1|k)}^{s}\right)\leq \frac{(N-1) b_c}{N} \leq b_c.$
This implies that at each time step we satisfy constraint~\eqref{eq:safeSetpoly} for the whole horizon $N$, i.e. $A_c(p_{i,(N|k)}-p_{i}(k))\leq b_c$. 
%Satisfaction of constraints~\eqref{seq:stateconstr},~\eqref{seq:termConstrEx} and~\eqref{seq:termConstrSafe} is guaranteed by the properties of the terminal sets stated in Assumption \ref{ass:cost_terminal} and since $\hat{\pazocal{S}}_i(p_i)\subset\mathbb{R}^2$.
Satisfaction of~\eqref{seq:couplingconstr} is instead guaranteed by the safe set $\hat{S}_i(p_i)$ and by the tightening~\eqref{eq:safeSetspeed}, but, due to the time-varying nature of the set $\pazocal{N}_{i}(k)$, at time $k+1$ there are four possibilities: \\
I) $\pazocal{N}_{i}(k+1)=\pazocal{N}_{i}(k)$, in this case the candidate solution satisfy the constraint thanks to~\eqref{seq:couplingconstr}. \\
II) Plug-out operation performed by a set of agents $\pazocal{A}$ i.e. one or more vehicles leaves the set of neighbouring systems for the agent $i$ and $\pazocal{N}_{i}(k+1)=\pazocal{N}_{i}(k)\backslash\pazocal{A}$.
In this case the obstacle avoidance constraint is automatically satisfied. \\
III) Plug-in operation performed by a set of agents $\pazocal{I}$ i.e. one or more vehicles are added to the set of neighbouring systems for the agent $i$, $\pazocal{N}_{i}(k+1)=\pazocal{N}_{i}(k) \cup \pazocal{I}$.
In this case, the safe trajectory of each agent computed at time $k$, satisfies constraint~\eqref{seq:couplingconstr} since $X_{i,k}^{s*}\in\pazocal{C}_{i,k}$ and $\pazocal{C}_{i,k}\cap \pazocal{C}_{j,k} = \emptyset$, representing a collision free trajectory also at time $k+1$.\\
    %tightening~\eqref{eq:safeSetspeed} ensures that if $\pazocal{C}_{i,k}\cap \pazocal{C}_{j,k} = \emptyset$ and $\pazocal{C}_{i,k+1}\cap \pazocal{C}_{j,k+1} \neq \emptyset$ then $\hat{\pazocal{S}}_{i,k+1}\cap \hat{\pazocal{S}}_{j,k+1} = \emptyset$, $\forall j\in \pazocal{I}$ ensuring that the trajectory~\eqref{eq:candidateSol} satisfies constraint~\eqref{seq:couplingconstr}.
IV) Finally when multiple plug-in, plug-out operations are performed, the two previous cases II) or III) can be considered subsequently.
Finally, satisfaction of constraint~\eqref{seq:convConstr} is trivial since the cost associated with the candidate trajectory~\eqref{eq:candidateSol} represents the right-hand side of inequality~\eqref{seq:convConstr} and it is an upper bound of the safe cost at the next time step.
\\
\textit{Asymptotic stability:} We consider standard arguments using Lyapunov stability theory, see e.g. \cite{mayne2000constrained}. In particular we consider the safe cost function~\eqref{eq:safe_cost} as a Lyapunov function. Let us denote with $J^{s,*}(x)$ the optimal safe cost computed by the FHOCP $\mathscr{P}(x,r,\hat{\pazocal{S}},\hat{J}^s)$ . Since the FHOCP~\eqref{eq:MPCproblem} is recursively feasible, the cost function $\hat{J}^{s}(x)$ computed with the candidate solution~\eqref{eq:candidateSol} represents an upper-bound of the value function at the next time step $x(k+1)=f(x(k),\kappa(x(k),r))$, i.e. $\hat{J}^{s}(x)\geq J^{s,*}(x(k+1))$.
Since the decreasing of the value function is imposed through constraint \eqref{seq:convConstr}, we have $J^{s,*}(x(k+1))\leq\hat{J}^{s}(x)\leq J^{s,*}(x)-l^s(x,\kappa(x,r))$. 
\end{proof}

\bibliographystyle{IEEEtran}
\bibliography{ifacconf}             % bib file to produce the bibliography

\begin{thebibliography}{10}
\providecommand{\url}[1]{#1}
\csname url@rmstyle\endcsname
\providecommand{\newblock}{\relax}
\providecommand{\bibinfo}[2]{#2}
\providecommand\BIBentrySTDinterwordspacing{\spaceskip=0pt\relax}
\providecommand\BIBentryALTinterwordstretchfactor{4}
\providecommand\BIBentryALTinterwordspacing{\spaceskip=\fontdimen2\font plus
\BIBentryALTinterwordstretchfactor\fontdimen3\font minus
  \fontdimen4\font\relax}
\providecommand\BIBforeignlanguage[2]{{%
\expandafter\ifx\csname l@#1\endcsname\relax
\typeout{** WARNING: IEEEtran.bst: No hyphenation pattern has been}%
\typeout{** loaded for the language `#1'. Using the pattern for}%
\typeout{** the default language instead.}%
\else
\language=\csname l@#1\endcsname
\fi
#2}}

\bibitem{siegwart2011introduction}
R.~Siegwart, I.~R. Nourbakhsh, and D.~Scaramuzza, \emph{Introduction to
  autonomous mobile robots}.\hskip 1em plus 0.5em minus 0.4em\relax MIT press,
  2011.

\bibitem{patil2020survey}
D.~Patil, M.~Ansari, D.~Tendulkar, R.~Bhatlekar, V.~N. Pawar, and S.~Aswale,
  ``A survey on autonomous military service robot,'' in \emph{2020
  International Conference on Emerging Trends in Information Technology and
  Engineering (ic-ETITE)}.\hskip 1em plus 0.5em minus 0.4em\relax IEEE, 2020,
  pp. 1--7.

\bibitem{bagloee2016autonomous}
S.~A. Bagloee, M.~Tavana, M.~Asadi, and T.~Oliver, ``Autonomous vehicles:
  challenges, opportunities, and future implications for transportation
  policies,'' \emph{Journal of modern transportation}, vol.~24, no.~4, pp.
  284--303, 2016.

\bibitem{shamma2008cooperative}
J.~Shamma, \emph{Cooperative control of distributed multi-agent systems}.\hskip
  1em plus 0.5em minus 0.4em\relax John Wiley \& Sons, 2008.

\bibitem{rawlings2017model}
J.~B. Rawlings, D.~Q. Mayne, and M.~Diehl, \emph{Model predictive control:
  theory, computation, and design}.\hskip 1em plus 0.5em minus 0.4em\relax Nob
  Hill Publishing Madison, WI, 2017, vol.~2.

\bibitem{stoustrup2009plug}
J.~Stoustrup, ``Plug \& play control: Control technology towards new
  challenges,'' \emph{European Journal of Control}, vol.~15, no. 3-4, pp.
  311--330, 2009.

\bibitem{saccani2021autonomous}
D.~Saccani and L.~Fagiano, ``Autonomous {UAV} navigation in an unknown
  environment via multi-trajectory model predictive control,'' in \emph{2021
  European Control Conference (ECC)}.\hskip 1em plus 0.5em minus 0.4em\relax
  IEEE, 2021, pp. 1577--1582.

\bibitem{limon2018nonlinear}
D.~Limon, A.~Ferramosca, I.~Alvarado, and T.~Alamo, ``Nonlinear {MPC} for
  tracking piece-wise constant reference signals,'' \emph{IEEE Transactions on
  Automatic Control}, vol.~63, no.~11, pp. 3735--3750, 2018.

\bibitem{fagiano2013generalized}
L.~Fagiano and A.~R. Teel, ``Generalized terminal state constraint for model
  predictive control,'' \emph{Automatica}, vol.~49, no.~9, pp. 2622--2631,
  2013.

\bibitem{wabersich2018safe}
K.~P. Wabersich and M.~N. Zeilinger, ``Safe exploration of nonlinear dynamical
  systems: A predictive safety filter for reinforcement learning,'' \emph{arXiv
  preprint arXiv:1812.05506}, 2018.

\bibitem{muntwiler2020distributed}
S.~Muntwiler, K.~P. Wabersich, A.~Carron, and M.~N. Zeilinger, ``Distributed
  model predictive safety certification for learning-based control,''
  \emph{IFAC-PapersOnLine}, vol.~53, no.~2, pp. 5258--5265, 2020.

\bibitem{ames2019control}
A.~D. Ames, S.~Coogan, M.~Egerstedt, G.~Notomista, K.~Sreenath, and P.~Tabuada,
  ``Control barrier functions: Theory and applications,'' in \emph{2019 18th
  European control conference (ECC)}.\hskip 1em plus 0.5em minus 0.4em\relax
  IEEE, 2019, pp. 3420--3431.

\bibitem{wabersich2022predictive}
K.~P. Wabersich and M.~N. Zeilinger, ``Predictive control barrier functions:
  Enhanced safety mechanisms for learning-based control,'' \emph{IEEE
  Transactions on Automatic Control}, 2022.

\bibitem{tordesillas2021faster}
J.~Tordesillas, B.~T. Lopez, M.~Everett, and J.~P. How, ``{FASTER}: Fast and
  safe trajectory planner for navigation in unknown environments,'' \emph{IEEE
  Transactions on Robotics}, 2021.

\bibitem{soloperto2022safe}
R.~Soloperto, A.~Mesbah, and F.~Allg{\"o}wer, ``Safe exploration and escape
  local minima with model predictive control under partially unknown
  constraints,'' \emph{arXiv preprint arXiv:2205.03614}, 2022.

\bibitem{saccani2022}
D.~Saccani, L.~Cecchin, and L.~Fagiano, ``Multitrajectory {M}odel {P}redictive
  {C}ontrol for {S}afe {UAV} {N}avigation in an {U}nknown {E}nvironment,''
  \emph{IEEE Transactions on Control Systems Technology}, 2022.

\bibitem{alsterda2021contingency}
J.~P. Alsterda and J.~C. Gerdes, ``Contingency model predictive control for
  linear time-varying systems,'' \emph{arXiv preprint arXiv:2102.12045}, 2021.

\bibitem{zeilinger2013plug}
M.~N. Zeilinger, Y.~Pu, S.~Riverso, G.~Ferrari-Trecate, and C.~N. Jones, ``Plug
  and play distributed model predictive control based on distributed invariance
  and optimization,'' in \emph{52nd IEEE conference on decision and
  control}.\hskip 1em plus 0.5em minus 0.4em\relax IEEE, 2013, pp. 5770--5776.

\bibitem{carron2021plug}
A.~Carron, K.~P. Wabersich, and M.~N. Zeilinger, ``Plug-and-play distributed
  safety verification for linear control systems with bounded uncertainties,''
  \emph{IEEE Transactions on Control of Network Systems}, vol.~8, no.~3, pp.
  1501--1512, 2021.

\bibitem{riverso2013plug}
S.~Riverso, M.~Farina, and G.~Ferrari-Trecate, ``Plug-and-play decentralized
  model predictive control for linear systems,'' \emph{IEEE Transactions on
  Automatic Control}, vol.~58, no.~10, pp. 2608--2614, 2013.

\bibitem{carron2023multi}
A.~Carron, D.~Saccani, L.~Fagiano, and M.~N. Zeilinger, ``Multi-agent
  distributed model predictive control with connectivity constraint,''
  \emph{arXiv preprint arXiv:2303.06957}, 2023.

\bibitem{gros2020linear}
S.~Gros, M.~Zanon, R.~Quirynen, A.~Bemporad, and M.~Diehl, ``From linear to
  nonlinear mpc: bridging the gap via the real-time iteration,''
  \emph{International Journal of Control}, vol.~93, no.~1, pp. 62--80, 2020.

\bibitem{engelmann2020decomposition}
A.~Engelmann, Y.~Jiang, B.~Houska, and T.~Faulwasser, ``Decomposition of
  nonconvex optimization via bi-level distributed {ALADIN},'' \emph{IEEE
  Transactions on Control of Network Systems}, vol.~7, no.~4, pp. 1848--1858,
  2020.

\bibitem{zhang2020optimization}
X.~Zhang, A.~Liniger, and F.~Borrelli, ``Optimization-based collision
  avoidance,'' \emph{IEEE Transactions on Control Systems Technology}, vol.~29,
  no.~3, pp. 972--983, 2020.

\bibitem{blanchini1999set}
F.~Blanchini, ``Set invariance in control,'' \emph{Automatica}, vol.~35,
  no.~11, pp. 1747--1767, 1999.

\bibitem{Andersson2019}
J.~A.~E. Andersson, J.~Gillis, G.~Horn, J.~B. Rawlings, and M.~Diehl,
  ``{CasADi} -- {A} software framework for nonlinear optimization and optimal
  control,'' \emph{Mathematical Programming Computation}, vol.~11, no.~1, pp.
  1--36, 2019.

\bibitem{wachter2006implementation}
A.~W{\"a}chter and L.~T. Biegler, ``On the implementation of an interior-point
  filter line-search algorithm for large-scale nonlinear programming,''
  \emph{Mathematical programming}, vol. 106, no.~1, pp. 25--57, 2006.

\bibitem{mayne2000constrained}
D.~Q. Mayne, J.~B. Rawlings, C.~V. Rao, and P.~O. Scokaert, ``Constrained model
  predictive control: Stability and optimality,'' \emph{Automatica}, vol.~36,
  no.~6, pp. 789--814, 2000.

\end{thebibliography}
\end{document}